\documentclass{article}

\usepackage{amsmath,amsthm,amssymb}

\usepackage{url}

\theoremstyle{theorem}
\newtheorem{theorem}{Theorem}

\newtheorem{lemma}[theorem]{Lemma}

\theoremstyle{definition}
\newtheorem{definition}[theorem]{Definition}
\newtheorem{example}[theorem]{Example}

\newtheorem{remark}[theorem]{Remark}

\begin{document}

\title{Constants of Motion for Non-Differentiable
Quantum Variational Problems\thanks{Partially presented at the
\emph{Fifth Symposium on Nonlinear Analysis} (SNA 2007),
Toru\'{n}, Poland, September 10-14, 2007. Accepted for publication (19/March/2008), \emph{Topological Methods in Nonlinear Analysis}, Lecture Notes of the Juliusz Schauder Center for Nonlinear Studies.}}

\author{Jacky Cresson~${}^{a}$\\ \texttt{jacky.cresson@univ-pau.fr}
        \and
        Gast\~{a}o S. F. Frederico~${}^{b, c}$\\ \texttt{gfrederico@mat.ua.pt}
        \and
        Delfim F. M. Torres~${}^{a, c}$\\ \texttt{delfim@ua.pt}}

\date{${}^{a}$~Laboratoire de Math\'{e}matiques Appliqu\'{e}es de Pau\\
      Universit\'{e} de Pau et des Pays de l'Adour\\
      Pau, France\\
      [0.3cm]
      ${}^{b}$~Department of Science and Technology\\
      University of Cape Verde\\
      Praia, Santiago, Cape Verde\\
      [0.3cm]
      ${}^{c}$~Centre for Research on Optimization and Control\\
      Department of Mathematics, University of Aveiro\\
      Aveiro, Portugal}

\maketitle


\begin{abstract}
We extend the DuBois-Reymond necessary optimality condition and
Noether's symmetry theorem to the scale relativity theory setting.
Both Lagrangian and Hamiltonian versions of Noether's theorem are
proved, covering problems of the calculus of variations with
functionals defined on sets of non-differentiable functions, as
well as more general non-differentiable problems of optimal
control. As an application we obtain constants of motion for some
linear and nonlinear variants of the Schr\"{o}dinger equation.

\medskip


\medskip


\end{abstract}


\section{Introduction}

The notion of \emph{symmetry} play an important role both in
physics and mathematics. Symmetries are defined as transformations
of a certain system, which result in the same object after the
transformation is carried out. They are mathematically described
by parameter groups of transformations. Their importance range
from fundamental and theoretical aspects to concrete applications,
having profound implications in the dynamical behavior of the
systems, and in their basic qualitative properties (see
\cite{torres:Kiev:04} and references therein).

\emph{Constants of motion} are another fundamental notion of
physics and mathematics. Typically, they are used in the calculus
of variations and optimal control to reduce the number of degrees
of freedom, thus reducing the problems to a lower dimension and
facilitating the integration of the equations given by the
necessary optimality conditions (see
\cite{GF:IJAM:06,EugenioDelfim06} and references therein).

Emmy Noether was the first to prove, in 1918, that these two
notions are connected: when a system exhibits a symmetry, then a
constant of motion exists. The celebrated \emph{Noether's theorem}
provide an explicit formula for such constants of motion. Since
the pioneer work of Emmy Noether, many extensions of the classical
results were done both in the calculus of variations setting as
well as in more general setting of optimal control (see
\cite{Cresson:TNS:2007,GF:IJTS:07,GF:JMAA:07,NonlinearDyn,CD:JMS:Torres:2002,CD:JMS:Torres:2004,GMJ06}
and references therein). All available versions of Noether's
theorem are, however, proved for problems whose admissible
functions are differentiable.

In 1992 L.~Nottale introduced the theory of scale-relativity
without the hypothesis of space-time differentiability
\cite{Nottale:1992,Nottale:1999}. A rigorous foundation to
Nottale's scale-relativity theory was recently given by
J.~Cresson \cite{CD:Cresson:2005,CD:Cresson:2006}.
The calculus of variations developed in \cite{CD:Cresson:2005} cover sets of non
differentiable curves, by substituting the classical derivative by
a new complex operator, known as the \emph{scale derivative}.

In this work we use the scale Euler-Lagrange equations and
respective scale extremals \cite{CD:Cresson:2005}, to prove an
extension of Noether's theorem for problems of the calculus of
variations and optimal control whose admissible functions are
non-differentiable (Theorems~\ref{theo:tnnd} and \ref{thm:NT:OC}).
The results are proved by first extending the classical
DuBois-Reymond necessary optimality condition to the scale
calculus of variations (Theorem~\ref{theo:cdrnd}). Illustrative
examples are given to Schr\"{o}dinger equations in the scale
framework \cite{AddaCresson:2004,AddaCresson:2005,CD:Cresson:2005}.


\section{Quantum Calculus}

In this section we briefly review the quantum calculus of
\cite{CD:Cresson:2005}, which extends the classical differential
calculus to non-differentiable functions.

We denote by $C^{0}$ the set of real-valued continuous functions
defined on $\mathbb{R}$.

\begin{definition}
Let $f\in C^{0}$. For all $\epsilon>0$, the $\epsilon$ left- and
right-quantum derivatives of $f$, denoted respectively by
$\Delta_{\epsilon}^{+}f(t)$ and $\Delta_{\epsilon}^{-}f(t)$, are
defined by
\begin{gather}
\Delta_{\epsilon}^{+}f(t)=\frac{f(t+\epsilon)-f(t)}{\epsilon}
\end{gather}
and
\begin{gather}
\Delta_{\epsilon}^{-}f(t)=\frac{f(t)-f(t-\epsilon)}{\epsilon} \, .
\end{gather}
\end{definition}

\begin{remark}
The $\epsilon$ left- and right-quantum derivative of a continuous
function $f$ correspond to the classical derivative of the
$\epsilon$-mean function $f_{\epsilon}^{\sigma}$ defined by
\begin{equation*}
f_{\epsilon}^{\sigma}(t)=\frac{\sigma}{\epsilon}\int_{t}^{t+\sigma\epsilon}f(s)ds\,
,
\end{equation*}
with $\sigma=\pm$.
\end{remark}

Next we define an operator which generalize the classical
derivative.

\begin{definition}
Let $f\in C^{0}$. For all $\epsilon>0$, the $\epsilon$ scale
derivative of $f$ at point $t$, denoted by
$\frac{\square_{\epsilon}f}{\square t}(t)$, is defined by
\begin{gather}
\frac{\square_{\epsilon}f}{\square
t}(t)=\frac{1}{2}\left[\left(\Delta_{\epsilon}^{+}f(t)+\Delta_{\epsilon}^{-}f(t)\right)
-i\left(\Delta_{\epsilon}^{+}f(t)-\Delta_{\epsilon}^{-}f(t)\right)\right]\,
.
\end{gather}
\end{definition}

\begin{remark}
If $f$ is differentiable, we can take the limit of the scale
derivative when $\epsilon$ goes to zero. We then obtain the
classical derivative $\frac{df}{dt}(t)$ of $f$ at $t$.
\end{remark}

We also need to extend the scale derivative to complex valued
functions.

\begin{definition}
 Let  $f$  be a continuous complex valued function. For all
$\epsilon>0$, the $\epsilon$ scale derivative of $f$, denoted by
$\frac{\square_{\epsilon}f}{\square t}$, is defined by
 \begin{gather}
\frac{{\square}_{\epsilon}f}{{\square}t}(t)=\frac{{\square}_{\epsilon}\textrm{Re}(f)}{\square
t}+i\frac{\square_{\epsilon}\textrm{Im}(f)}{\square t} \, ,
\end{gather}
where $\textrm{Re}(f)$ and $\textrm{Im}(f)$ denote the real and
imaginary part of $f$ respectively.
\end{definition}

In what follows, we will frequently use $\square_{\epsilon}$ to
denote the scale derivative operator
$\frac{{\square}_{\epsilon}}{{\square}t}$.

\begin{theorem}[\textrm{cf.} \cite{CD:Cresson:2005}]
\label{theo:mult} Let $f$ and $ g$ be two $C^{0}$ functions. For
all $\epsilon>0$ one has
\begin{equation}
\label{eq:mult} \square_{\epsilon}(f\cdot
g)=\square_{\epsilon}f\cdot g+f\cdot\square_{\epsilon}g +\epsilon
i\left(\square_{\epsilon}f\boxminus_{\epsilon}g-\boxminus_{\epsilon}f\square_{\epsilon}g
-\square_{\epsilon}f\square_{\epsilon}g-\boxminus_{\epsilon}f\boxminus_{\epsilon}g\right)
\end{equation}
where $\boxminus f$ denotes the complex conjugate of \ $\square f$.
\end{theorem}

\begin{remark}
For two differentiable functions $f$ and $g$, one obtains the
classical Leibniz rule $(f\cdot g)'=f'\cdot g+f\cdot g'$ by
taking the limit of \eqref{eq:mult} when $\epsilon$ goes to
zero.
\end{remark}

\begin{remark}
It is not difficult to prove the following equality:
\begin{equation}
\label{eq:int}
\int_{a}^{b}\square_{\epsilon}f(t)dt=\left.\frac{1}{2}\left[\left(f_{\epsilon}^{+}(t)+f_{\epsilon}^{-}(t)\right)
-i\left(f_{\epsilon}^{+}(t)-f_{\epsilon}^{-}(t)\right)\right]\right|_{a}^{b}\, .
\end{equation}
When $\epsilon$ goes to zero, \eqref{eq:int} reduces to
\begin{equation*}
\int_{a}^{b}\frac{d}{dt}f(t) dt = \left.f(t)\right|_{a}^{b}\, .
\end{equation*}
\end{remark}

\begin{definition}[$\alpha$-H\"{o}lderian functions]
A continuous real valued function $f$ is said to be $\alpha$-H\"{o}lderian,
$0<\alpha<1$, if for all $\epsilon>0$ and
all $t$, $t'\in\mathbb{R}$ there exists a constant $c$ such that
$|t-t'|\leqslant\epsilon$ implies $|f(t)-f(t')|\leqslant c\epsilon^{\alpha}$.
\end{definition}

We denote by $H^{\alpha}$ the set of continuous
functions which are $\alpha$-H\"{o}lderian.

\begin{theorem}[\textrm{cf.} \cite{CD:Cresson:2005}]
\label{theo:derivada} Let $f(t,x)$ be a $C^{n+1}$ real valued function and
$x(t)\in H^{1/n}$, $n\geq 1$. For all $\epsilon > 0$ sufficiently
small one has
\begin{equation}
\label{eq:derivada}
\frac{{\square}_{\epsilon}f}{{\square}t}(t,x(t))=\frac{\partial
f}{\partial t}(t,x(t))+\sum^{n}_{j=1}\frac{1}{j!}\frac{\partial
^{j}f}{\partial
x^{j}}\left(t,x(t)\right)\epsilon^{j-1}a_{\epsilon,j}(t)+o\left(\epsilon^{1/n}\right)
\end{equation}
where
\begin{equation*}
a_{\epsilon,j}(t)=\frac{1}{2}\left[\left(\left(\Delta^{\epsilon}_{+}x\right)^{j}-(-1)^{j}
\left(\Delta^{\epsilon}_{-}x\right)^{j}\right)-i\left(\left(\Delta^{\epsilon}_{+}x\right)^{j}+(-1)^{j}
\left(\Delta^{\epsilon}_{-}x\right)^{j}\right)\right]\, .
\end{equation*}
\end{theorem}

Lemma~\ref{lem:lema} is crucial for our purposes
(see proof of  Theorem~\ref{theo:tnnd}).

\begin{lemma}[\textrm{cf.} \cite{CD:Cresson:2005}]
\label{lem:lema} Let $h\in H^{\beta}$, $\beta\geq
\alpha 1_{[1/2,1]}+(1-\alpha)1_{]0,1/2[}$,
satisfy $h(a)=h(b)=0$ for some $a,b\in
\mathbb{R}$. If
$f_{\epsilon}:\mathbb{R}\longmapsto \mathbb{C}$,
$\epsilon>0$, is such that for all $t\in [a,b]$ one has
\begin{equation*}
\sup_{s\in
\left\{{{t,t+\sigma\epsilon}}\right\}}|f_{\epsilon}(s)|\leq
C\epsilon^{\alpha-1} \, ,
\end{equation*}
then
\begin{equation*}
\int_a^b\frac{\square_\epsilon}{\square t}\left(f_\epsilon
(t)h(t)\right)dt=o\left(\epsilon^{\alpha+\beta-1}\right)
\end{equation*}
and
\begin{equation*}
\epsilon\int_a^b Op_\epsilon(f_\epsilon)Op'_\epsilon(h)dt=o\left(\epsilon^{\alpha+\beta}\right)
\end{equation*}
where $Op_\epsilon$ and $Op'_\epsilon$ are either
$\square_\epsilon$ or $\boxminus_\epsilon$.
\end{lemma}


\section{Review of the Classical Noether's Theorem}
\label{sec:cNT}

There are several ways to prove the classical Noether's theorem.
In this section we review one of those proofs.

We begin by formulating the fundamental problem of the calculus of variations:
to minimize
\begin{equation}
\label{P} I[q(\cdot)] = \int_a^b L\left(t,q(t),\dot{q}(t)\right) dt
\end{equation}
under given boundary conditions $q(a)=q_{a}$ and $q(b)=q_{b}$, and
where $\dot{q} = \frac{dq}{dt}$. The Lagrangian $L :[a,b] \times
\mathbb{R}^{n} \times \mathbb{R}^{n} \rightarrow \mathbb{R}$ is
assumed to be a $C^{1}$-function with respect to all its arguments,
and admissible functions $q(\cdot)$ are assumed to be $C^2$-smooth.

\begin{definition}[Invariance of \eqref{P}]
\label{def:inva} The functional \eqref{P} is said to be invariant
under the $s$-parameter group of infinitesimal transformations
\begin{equation}
\label{eq:tinf}
\begin{cases}
\bar{t} = t + s\tau(t,q) + o(s) \, ,\\
\bar{q}(t) = q(t) + s\xi(t,q) + o(s) \, ,\\
\end{cases}
\end{equation}
if
\begin{equation}
\label{eq:inv1}  \int_{t_{a}}^{t_{b}}
L\left(t,q(t),\dot{q}(t)\right)dt =
\int_{\bar{t}(t_a)}^{\bar{t}(t_b)}
L\left(\bar{t},\bar{q}(\bar{t}),\dot{\bar{q}}(\bar{t})\right)d\bar{t}
\end{equation}
for any subinterval $[{t_{a}},{t_{b}}] \subseteq [a,b]$.
\end{definition}

We will denote by $\partial_{i}L$ the partial derivative of $L$ with
respect to its $i$-th argument, $i = 1,2,3$.

\begin{theorem}[Necessary and sufficient condition of invariance]
\label{theo:cnsi} If functional \eqref{P} is invariant under
transformations \eqref{eq:tinf}, then
\begin{equation}
\label{eq:cnsi}
\begin{split}
\partial_{1}
&L\left(t,q,\dot{q}\right)\tau+\partial_{2}
L\left(t,q,\dot{q}\right)\cdot\xi  \\
&+\partial_{3}
L\left(t,q,\dot{q}\right)\cdot\left(\dot{\xi}-\dot{q}\dot{\tau}\right)
+L\left(t,q,\dot{q}\right)\dot{\tau}=0 \, .
\end{split}
\end{equation}
\end{theorem}

\begin{proof}
Since \eqref{eq:inv1} is to be satisfied for any
subinterval $[{t_{a}},{t_{b}}]$ of $[a,b]$, one can get rid off
of the integral sign in \eqref{eq:inv1} and write the
equivalent equality
\begin{equation}
\label{eq:inv2} L\left(t,q,\dot{q}\right) =\left[
L\left(t+s\tau+o(s),q+s\xi+o(s), \frac{\dot{q}+
s\dot{\xi}+o(s)}{1+s
\dot{\tau}+o(s)}\right)\right]\frac{d\bar{t}}{dt} \, .
\end{equation}
Equation \eqref{eq:cnsi} is obtained differentiating both sides of
condition \eqref{eq:inv2} with respect to $s$ and then putting
$s=0$.
\end{proof}

\begin{definition}[Constant of motion]
A quantity $C(t,q(t),\dot{q}(t))$, $t \in [a,b]$, is said to be a
\emph{constant of motion} if
$\frac{d}{dt}C(t,q(t),\dot{q}(t))=0$ for all the solutions $q$ of
the Euler-Lagrange equation
\begin{equation}
\label{eq:el} \frac{d}{{dt}}\partial_{3} L\left(t,q(t),\dot{q}(t)\right) =
\partial_{2} L\left(t,q(t),\dot{q}(t)\right) \, .
\end{equation}
\end{definition}

\begin{theorem}[DuBois-Reymond necessary optimality condition]
\label{theo:cdr} If function $q$ is a minimizer or maximizer of
functional (\ref{P}), then
\begin{equation}
\label{eq:cdr}
\partial_{1}
L\left(t,q(t),\dot{q}(t)\right)=\frac{d}{dt}\left\{L\left(t,q(t),\dot{q}(t)\right)
-\partial_{3} L\left(t,q(t),\dot{q}(t)\right)\cdot\dot{q}(t)\right\}\, .
\end{equation}
\end{theorem}

\begin{proof}
The conclusion follows by direct calculations using the Euler-Lagrange
equation \eqref{eq:el}:
\begin{equation*}
\begin{split}
\frac{d}{dt} \{L\left(t,q,\dot{q}\right)&-\partial_{3}
L\left(t,q,\dot{q}\right)\cdot\dot{q}\}\\
&=\partial_{1} L\left(t,q,\dot{q}\right) +\partial_{2}
L\left(t,q,\dot{q}\right)\cdot\dot{q} +\partial_{3}
L\left(t,q,\dot{q}\right)\cdot\ddot{q} \\
& \qquad -\frac{d}{dt}\partial_{3}
L\left(t,q,\dot{q}\right)\cdot\dot{q}-\partial_{3}
L\left(t,q,\dot{q}\right)\cdot\ddot{q}\\
&=\partial_{1} L\left(t,q,\dot{q}\right)+\dot{q}\cdot(\partial_{2}
L\left(t,q,\dot{q}\right)-\frac{d}{dt}\partial_{3}
L\left(t,q,\dot{q}\right))\\
&=\partial_{1} L\left(t,q,\dot{q}\right)\, .
\end{split}
\end{equation*}
\end{proof}

\begin{theorem}[Noether's theorem]
\label{theo:tnoe} If \eqref{P} is invariant under
\eqref{eq:tinf}, then
\begin{equation}
\label{eq:TeoNet} C(t,q,\dot{q}) =
\partial_{3} L\left(t,q,\dot{q}\right)\cdot\xi(t,q)
+ \left( L(t,q,\dot{q}) - \partial_{3} L\left(t,q,\dot{q}\right)
\cdot \dot{q} \right) \tau(t,q)
\end{equation}
is a constant of motion.
\end{theorem}

\begin{proof}
To prove Noether's theorem we use the
Euler-Lagrange equations \eqref{eq:el} and the
DuBois-Reymond necessary optimality
condition \eqref{eq:cdr} into the necessary and sufficient
condition of invariance \eqref{eq:cnsi}:
\begin{equation*}
\begin{split}
0 &= \partial_{1} L\left(t,q,\dot{q}\right)\tau+\partial_{2} L\left(t,q,\dot{q}\right)\cdot\xi
+\partial_{3}
L\left(t,q,\dot{q}\right)\cdot\left(\dot{\xi}-\dot{q}\dot{\tau}\right)
+L\left(t,q,\dot{q}\right)\dot{\tau} \\
&=\partial_{2} L\left(t,q,\dot{q}\right)\cdot\xi+\partial_{3}
L\left(t,q,\dot{q}\right)\cdot\dot{\xi}+\partial_{1}
L\left(t,q,\dot{q}\right)\tau\\
&\qquad +\dot{\tau}\left(L\left(t,q,\dot{q}\right)-\partial_{3}
L\left(t,q,\dot{q}\right)\cdot\dot{q}\right)\\
&=\frac{d}{{dt}}\partial_{3}
L\left(t,q,\dot{q}\right)\cdot\xi+\partial_{3}
L\left(t,q,\dot{q}\right)\cdot\dot{\xi}+\frac{d}{dt}\left\{L\left(t,q,\dot{q}\right)-\partial_{3}
L\left(t,q,\dot{q}\right)\cdot\dot{q}\right\}\tau\\
&\qquad +\dot{\tau}\left(L\left(t,q,\dot{q}\right)-\partial_{3}
L\left(t,q,\dot{q}\right)\cdot\dot{q}\right)\\
&=\frac{d}{dt}\left\{\partial_{3} L\left(t,q,\dot{q}\right)\cdot\xi +
\left( L(t,q,\dot{q}) -
\partial_{3} L\left(t,q,\dot{q}\right) \cdot \dot{q} \right)
\tau\right\}\, .
\end{split}
\end{equation*}
\end{proof}


\section{Main Results: Non-differentiable Noether-type Theorems}

The classical Noether's theorem is valid
for extremals $q(\cdot)$ which are $C^2$ differentiable,
as considered in Section~\ref{sec:cNT}. The biggest
class where a Noether-type theorem has been proved is
the class of Lipschitz functions \cite{CD:JMS:Torres:2004}.
In this work we prove a more general Noether-type theorem, valid for
non-differentiable scale extremals.


\subsection{Calculus of variations with scale derivatives}

In \cite{CD:Cresson:2005} the calculus of variations with scale
derivatives is introduced and respective Euler-Lagrange equations
derived. In this section we obtain a formulation of Noether's
theorem for the scale calculus of variations. The proof of our
Noether's theorem is done in two steps: first we extend the
DuBois-Reymond condition to problems with scale derivatives
(Theorem~\ref{theo:cdrnd}); then, using this result, we obtain the
scale/quantum Noether's theorem (Theorem~\ref{theo:tnnd}).

The problem of the calculus of variations with scale derivatives
is defined as
\begin{gather}
\label{Pe} I[q(\cdot)] = \int_{a-\epsilon}^{b+\epsilon}
L\left(t,q(t),\square_{\epsilon}q(t)\right) dt \longrightarrow \min
\end{gather}
under given boundary conditions $q(a-\epsilon)=q_{a_\epsilon}$
and $q(b+\epsilon)=q_{b_\epsilon}$, $0<\epsilon\ll 1$,
$q(\cdot) \in H^{\alpha}$, $0<\alpha<1$.
The Lagrangian $L :[a-\epsilon,b+\epsilon] \times
\mathbb{R}^{n}\times \mathbb{C}^{n} \rightarrow \mathbb{C}$ is
assumed to be a $C^{1}$-function with respect to all its arguments
satisfying
\begin{equation}
\label{eq:ctnnd}
 ||DL(t,q(t),\square_{\epsilon}q(t))||\leq K \, ,
\end{equation}
where $K$ is a non-negative constant, $D$ denotes the differential
and $||\cdot||$ is a norm for matrices (see \cite{CD:Cresson:2005}).

\begin{remark}
In the case of admissible differentiable functions $q(\cdot)$,
problem~\eqref{Pe} tends to problem \eqref{P} when $\epsilon$
tends to zero.
\end{remark}

\begin{remark}
We need to assume $a-\epsilon\leq t\leq b+\epsilon$ in order to
avoid problems with the definition of scale derivative in the
boundaries of the interval.
\end{remark}

\begin{theorem}[Scale Euler-Lagrange equation -- \textrm{cf.} \cite{CD:Cresson:2005}]
\label{Thm:NonDtELeq} If $q$ is a minimizer of problem \eqref{Pe},
then $q$ satisfy the following \emph{scale Euler-Lagrange
equation}:
\begin{equation}
\label{eq:elnd}
\partial_{2} L\left(t,q(t),\square_{\epsilon}q(t)\right)-\square_{\epsilon}
\partial_{3} L\left(t,q(t),\square_{\epsilon}q(t)\right)=0\, .
\end{equation}
\end{theorem}

\begin{definition}[Scale extremals]
\label{def:scale:ext}
The solutions $q(t)$ of the scale Euler-Lagrange
equation \eqref{eq:elnd} are called \emph{scale extremals}.
\end{definition}

\begin{definition}[\textrm{cf.} Definition~\ref{def:inva}]
\label{def:invnd} The functional \eqref{Pe} is said to be
invariant under a $s$-parameter group of infinitesimal
transformations
\begin{equation}
\label{eq:tinf:sd}
\begin{cases}
\bar{t} = t + s\tau(t,q) + o(s) \, ,\\
\bar{q}(t) = q(t) + s\xi(t,q) + o(s) \, ,\\
\end{cases}
\end{equation}
$\tau$, $\xi\in H^{\beta}$, $\beta\geq \alpha
1_{[1/2,1]}+(1-\alpha)1_{]0,1/2[}$, if
\begin{equation}
\label{eq:invnd}
\int_{t_{a}}^{t_{b}}L\left(t,q(t),\square_{\epsilon}q(t)\right)dt=\int_{\bar{t}(t_a)}^{\bar{t}(t_b)}
L\left(\bar{t},\bar{q}(\bar{t}),\square_{\epsilon}\bar{q}(\bar{t})\right)dt
\end{equation}
for any subinterval $[{t_{a}},{t_{b}}] \subseteq
[a-\epsilon,b+\epsilon]$.
\end{definition}

Theorem~\ref{thm:CNSI:SCV} establish a necessary and sufficient
condition of invariance for \eqref{Pe}. Condition
\eqref{eq:cnsind} will be used in the proof of our Noether-type
theorem.

\begin{theorem}[\textrm{cf.} Theorem~\ref{theo:cnsi}]
\label{thm:CNSI:SCV} If functional \eqref{Pe} is invariant under
the one-parameter group of transformations \eqref{eq:tinf:sd},
then
\begin{multline}
\label{eq:cnsind} \int_{t_{a}}^{t_{b}}\Bigl[\partial_{1}
L\left(t,q(t),\square_{\epsilon}q(t)\right)\tau +\partial_{2}
L\left(t,q(t),\square_{\epsilon}q(t)\right)\cdot\xi
\\
+\partial_{3}
L\left(t,q(t),\square_{\epsilon}q(t)\right)\cdot\left(\square_{\epsilon}\xi
-\square_{\epsilon}q(t)\square_{\epsilon}\tau\right)\Bigr]dt
 = 0
\end{multline}
for any subinterval $[{t_{a}},{t_{b}}] \subseteq
[a-\epsilon,b+\epsilon]$.
\end{theorem}

\begin{proof}
Equation \eqref{eq:invnd} is equivalent to
\begin{multline}
\label{eq:invndd}
\int_{t_{a}}^{t_{b}}L\left(t,q(t),\square_{\epsilon}q(t)\right)dt\\
= \int_{t_a+s\tau}^{t_b+s\tau} L\left(t+s\tau+o(s),q+s\xi+o(s)
,\frac{\square_{\epsilon}q+s\square_{\epsilon}\xi+o(s)}
{1+s\square_{\epsilon}\tau+o(s)}\right)dt\, .
\end{multline}
Differentiating both sides of equation \eqref{eq:invndd} with
respect to $s$, then putting $s=0$, we obtain equality
\eqref{eq:cnsind}.
\end{proof}

\begin{definition}[Scale constants of motion]
\label{def:leicond} We say that quantity
$C(t,q(t),\square_{\epsilon}q(t))$ is a \emph{scale constant of
motion} if, and only if,
$C(t,q(t),\square_{\epsilon}q(t))=constant$ along all the
scale extremals $q(\cdot)$ (\textrm{cf.} Definition~\ref{def:scale:ext}).
\end{definition}

Theorem~\ref{theo:cdrnd} generalizes the DuBois-Reymond necessary
optimality condition (\textrm{cf.} Theorem~\ref{theo:cdr}) for
problems of the calculus of variations with scale derivatives.

\begin{theorem}[Scale DuBois-Reymond necessary condition]
\label{theo:cdrnd} If $q(\cdot)$ is a minimizer of problem
(\ref{Pe}), then it satisfy the following condition:
\begin{multline}
\label{eq:cdrnd} \frac{\square_{\epsilon}}{\square
t}\left\{L\left(t,q,\frac{\square_{\epsilon}q}{\square
t}\right)-\partial_{3} L\left(t,q,\frac{\square_{\epsilon}q}{\square
t}\right)\cdot\frac{\square_{\epsilon}q}{\square t}\right\} \\=
\partial_{1} L\left(t,q,\frac{\square_{\epsilon}q}{\square
t}\right)-\epsilon
i\left(\square_{\epsilon}f\boxminus_{\epsilon}g-\boxminus_{\epsilon}f\square_{\epsilon}g
-\square_{\epsilon}f\square_{\epsilon}g-\boxminus_{\epsilon}f\boxminus_{\epsilon}g\right)
\,
\end{multline}
where $f=\partial_{3} L\left(t,q,\frac{\square_{\epsilon}q}{\square
t}\right)$, $g=\frac{\square_{\epsilon}q}{\square t}$ , $\boxminus
f$ and $\boxminus g$ are the complex conjugate of $\square f$ and
$\square g$ respectively.
\end{theorem}

\begin{proof}
The scale DuBois-Reymond necessary condition \eqref{eq:cdrnd}
follows from the linearity of the scale derivative operator,
Theorems~\ref{theo:mult} and \ref{theo:derivada}, and the scale
Euler-Lagrange equations \eqref{eq:elnd}:
\begin{equation*}
\begin{split}
\square_{\epsilon}\Bigl\{L(t,&q,\square_{\epsilon}q)-\partial_{3}
L(t,q,\square_{\epsilon}q)\cdot\square_{\epsilon}q\Bigr\} \\
&=\partial_{1}L(t,q,\square_{\epsilon}q)
+\partial_{2}L(t,q,\square_{\epsilon}q)\cdot\square_{\epsilon}q
+\partial_{3}L(t,q,\square_{\epsilon}q)\cdot\square_{\epsilon}\square_{\epsilon}q \\
&\qquad-\square_{\epsilon}\partial_{3}L(t,q,\square_{\epsilon}q)\cdot\square_{\epsilon}q
-\partial_{3}L(t,q,\square_{\epsilon}q)\cdot\square_{\epsilon}\square_{\epsilon}q \\
&\qquad-\epsilon
i\left(\square_{\epsilon}f\boxminus_{\epsilon}g-\boxminus_{\epsilon}f\square_{\epsilon}g
-\square_{\epsilon}f\square_{\epsilon}g-\boxminus_{\epsilon}f\boxminus_{\epsilon}g\right)\\
&=\partial_{1}L(t,q,\square_{\epsilon}q)+\square_{\epsilon}q\cdot(\partial_{2}L(t,q,\square_{\epsilon}q)
-\square_{\epsilon}\partial_{3}L(t,q,\square_{\epsilon}q)) \\
&\qquad -\epsilon
i\left(\square_{\epsilon}f\boxminus_{\epsilon}g-\boxminus_{\epsilon}f\square_{\epsilon}g
-\square_{\epsilon}f\square_{\epsilon}g-\boxminus_{\epsilon}f\boxminus_{\epsilon}g\right)\\
&=\partial_{1}L(t,q,\square_{\epsilon}q)-\epsilon
i\left(\square_{\epsilon}f\boxminus_{\epsilon}g-\boxminus_{\epsilon}f\square_{\epsilon}g
-\square_{\epsilon}f\square_{\epsilon}g-\boxminus_{\epsilon}f\boxminus_{\epsilon}g\right)
\, .
\end{split}
\end{equation*}
\end{proof}

Theorem~\ref{theo:tnnd} establish an extension of Noether's
theorem for problems of the calculus of variations with scale
derivatives.

\begin{theorem}[Scale Noether's theorem in Lagrangian form]
\label{theo:tnnd} If functional \eqref{Pe} is invariant in the
sense of Definition~\ref{def:invnd}, then
\begin{multline}
\label{eq:tnnd} C(t,q(t),\square_{\epsilon}q(t)) =
\partial_{3}L(t,q,\square_{\epsilon}q))\cdot\xi(t,q)\\
+\Bigl(L(t,q,\square_{\epsilon}q)-\partial_{3}
L(t,q,\square_{\epsilon}q)\cdot\square_{\epsilon}q\Bigr)\tau(t,q)
\end{multline}
is a scale constant of motion (\textrm{cf.}
Definition~\ref{def:leicond}).
\end{theorem}

\begin{remark}
If the admissible functions $q$ are differentiable, the scale
constant of motion \eqref{eq:tnnd} tends to \eqref{eq:TeoNet} when
we take the limit $\epsilon \rightarrow 0$.
\end{remark}

\begin{proof}
Noether's scale constant of motion \eqref{eq:tnnd} follows by
using the scale DuBois-Raymond condition \eqref{eq:cdrnd}, the
scale Euler-Lagrange equation \eqref{eq:elnd} and
Theorem~\ref{theo:mult} into the necessary and sufficient
condition of invariance \eqref{eq:cnsind}:
\begin{equation}
\label{eq:dtnnd}
\begin{split}
0&=\int_{t_{a}}^{t_{b}}\Bigl[\partial_{1}
L\left(t,q(t),\square_{\epsilon}q(t)\right)\tau +\partial_{2}
L\left(t,q(t),\square_{\epsilon}q(t)\right)\cdot\xi \\
& \qquad\qquad +\partial_{3}
L\left(t,q,\square_{\epsilon}q\right)\cdot\left(\square_{\epsilon}\xi
-\square_{\epsilon}q\square_{\epsilon}\tau\right)
+L\square_{\epsilon}\tau-L\square_{\epsilon}\tau\Bigr]dt\\
&=\int_{t_{a}}^{t_{b}}\Bigl[\tau\square_{\epsilon}(L\left(t,q,\square_{\epsilon}q\right)-\partial_{3}
L\left(t,q,\square_{\epsilon}q\right)\cdot\square_{\epsilon}q)\\
&\qquad\qquad +\left(L\left(t,q,\square_{\epsilon}q\right)
-\partial_{3}
L\left(t,q,\square_{\epsilon}q\right)\cdot\square_{\epsilon}q\right)\square_{\epsilon}\tau \\
&\qquad\qquad +\xi\cdot\square_{\epsilon}\partial_{3}
L\left(t,q,\square_{\epsilon}q\right) +\partial_{3}
L\left(t,q,\square_{\epsilon}q\right)\cdot\square_{\epsilon}\xi\Bigr]dt+R(\epsilon)\\
&=\int_{t_{a}}^{t_{b}}\frac{\square_{\epsilon}}{\square
t}\Big\{\partial_{3}
L\left(t,q,\square_{\epsilon}q\right)\cdot\xi+(L\left(t,q,\square_{\epsilon}q\right)-\partial_{3}
L\left(t,q,\square_{\epsilon}q\right)\cdot\square_{\epsilon}q)\tau\Big\}dt\\
&\qquad +R(\epsilon) +R'(\epsilon)\, ,
\end{split}
\end{equation}
where $R(\epsilon)$ and $R'(\epsilon)$ are integrals with terms
resulting from the application of formula \eqref{eq:mult} of
Theorem~\ref{theo:mult}. Taking into consideration condition
\eqref{eq:ctnnd} and Lemma~\ref{lem:lema}, the integrals
$R(\epsilon)$ and $R'(\epsilon)$ vanish. Thus, \eqref{eq:dtnnd}
simplify to
\begin{equation}
\label{eq:dtnnd1}
\int_{t_{a}}^{t_{b}}\frac{\square_{\epsilon}}{\square
t}\Big\{\partial_{3}
L\left(t,q,\square_{\epsilon}q\right)\cdot\xi+(L\left(t,q,\square_{\epsilon}q\right)-\partial_{3}
L\left(t,q,\square_{\epsilon}q\right)\cdot\square_{\epsilon}q)\tau\Big\}dt
=0 \, .
\end{equation}
Using formula \eqref{eq:int} and having in mind that
\eqref{eq:dtnnd1} holds for an arbitrary $[{t_{a}},{t_{b}}]
\subseteq [a-\epsilon,b+\epsilon]$, we conclude that
\begin{equation*}
L\left(t,q,\square_{\epsilon}q\right)\cdot\xi+\Bigl(L\left(t,q,\square_{\epsilon}q\right)-\partial_{3}
L\left(t,q,\square_{\epsilon}q\right)\cdot\square_{\epsilon}q\Bigr)\tau
=constant \, .
\end{equation*}
\end{proof}


\subsection{Scale optimal control}

Theorem~\ref{theo:tnnd} gives a Lagrangian formulation
of Noether's principle to the non-differentiable scale setting.
Now we give a scale Hamiltonian formulation of Noether's principle
for more general scale problems of optimal control (Theorem~\ref{thm:NT:OC}).
The result is obtained as a corollary of Theorem~\ref{theo:tnnd}.

We define the scale optimal control problem as follows:
\begin{gather}
\label{Pond} I[q(\cdot),u(\cdot)] = \int_{a-\epsilon}^{b+\epsilon}
L\left(t,q(t),u(t)\right) dt \longrightarrow \min \, , \\
\square_{\epsilon}{q}(t)=\varphi\left(t,q(t),u(t)\right) \, ,
\notag
\end{gather}
under the given initial condition $q(a-\epsilon)=q_{a_\epsilon}$,
$0<\epsilon\ll 1$, $q(\cdot)$, $u(\cdot) \in H^{\alpha}$, $0<\alpha<1$. The
Lagrangian $L :[a-\epsilon,b+\epsilon] \times \mathbb{R}^{n}\times
\mathbb{C}^{m} \rightarrow \mathbb{C}$ and the velocity vector
$\varphi : [a-\epsilon,b+\epsilon] \times \mathbb{R}^{n}\times
\mathbb{C}^m\rightarrow \mathbb{C}^{n}$ are assumed to be
$C^{1}$-functions with respect to all its arguments. Similarly as
before, we assume that
\begin{equation*}
||DL(t,q(t),u(t))|| \leq K \, ,
\end{equation*}
where $K$ is a non-negative constant, $D$ denotes the differential
and $||\cdot||$ a classical norm of matrices.

\begin{remark}
In the particular case when $\varphi(t,q,u) = u$, \eqref{Pond} is
reduced to the scale problem of the calculus of variations
\eqref{Pe}.
\end{remark}

\begin{remark}
If functions $q$ are differentiable, problem \eqref{Pond} tends to
the classical problem of optimal control,
\begin{gather}
\label{eq:JO1}
I[q(\cdot),u(\cdot)] =\int_a^b L\left(t,q(t),u(t)\right) dt \longrightarrow \min \, ,\\
\dot{q}(t)=\varphi\left(t,q(t),u(t)\right) \, , \notag
\end{gather}
as $\epsilon \rightarrow 0$.
\end{remark}

\begin{theorem}
\label{theo:pmpnd} If $(q(\cdot),u(\cdot))$ is a minimizer of
\eqref{Pond}, then there exists a co-vector function $p(t)\in
H^{\alpha}([a-\epsilon,b+\epsilon];\mathbb{R}^{n})$ such that the
following conditions hold:
\begin{itemize}
\item the scale Hamiltonian system
\begin{equation}
\label{eq:Hamnd}
\begin{cases}
\square_{\epsilon}{q}(t)&=\partial_{4} {\cal H}(t, q(t), u(t),p(t)) \, , \\
\square_{\epsilon}{p}(t) &=-\partial_{2} {\cal H}(t, q(t), u(t),
p(t)) \, ;
\end{cases}
\end{equation}
\item the stationary condition
\begin{equation}
\label{eq:CE}
 \partial_{3} {\cal H}(t, q(t), u(t), p(t))=0 \, ;
\end{equation}
\end{itemize}
where the Hamiltonian ${\cal H}$ is defined by
\begin{equation}
\label{eq:Hnd} {\cal H}\left(t,q,u,p\right) =L\left(t,q,u\right)+p
\cdot \varphi\left(t,q,u\right) \, .
\end{equation}
\end{theorem}

\begin{remark}
The first equation in the scale Hamiltonian system
\eqref{eq:Hamnd} is nothing more than the scale control system
$\square_{\epsilon}{q}(t)=\varphi\left(t,q(t),u(t)\right)$ given
in the formulation of problem \eqref{Pond}.
\end{remark}

\begin{remark}
In classical mechanics $p$ is called the \emph{generalized
momentum}. In the language of optimal control
\cite{CD:MR29:3316b}, $p$ is known as the adjoint variable.
\end{remark}

\begin{definition}
\label{scale:Pont:Ext} A triplet $(q(\cdot),u(\cdot),p(\cdot))$
satisfying the conditions of Theorem~\ref{theo:pmpnd} will be
called a \emph{scale Pontryagin extremal}.
\end{definition}

\begin{remark}
\label{rem:cp:CV:EP:EEL} In the particular case when
$\varphi(t,q,u) = u$, Theorem~\ref{theo:pmpnd} reduces to
Theorem~\ref{Thm:NonDtELeq}: the stationary condition
\eqref{eq:CE} gives $p = - \partial_3 L$ and the second equation
in the scale Hamiltonian system \eqref{eq:Hamnd} gives
$\square_{\epsilon}{p}(t) =-\partial_{2} L$. Comparing both
equalities, one obtains the scale Euler-Lagrange equation
\eqref{eq:elnd}: $\square_{\epsilon}
\partial_3 L = \partial_{2} L$. In other words, the scale
Pontryagin extremals (Definition~\ref{scale:Pont:Ext}) are a
generalization of the scale Euler-Lagrange extremals
(Definition~\ref{def:scale:ext}).
\end{remark}

\begin{proof} (of Theorem~\ref{theo:pmpnd})
Using the Lagrange multiplier rule, \eqref{Pond} is equivalent to
the augmented problem
\begin{equation}
\label{eq:pcond} J[q(\cdot),u(\cdot),p(\cdot)] = \int_a^b
\left[{\cal
H}\left(t,q(t),u(t),p(t)\right)-p(t)\cdot\square_{\epsilon}q(t)\right]dt
\longrightarrow \min \, .
\end{equation}
The necessary optimality condition \eqref{eq:Hamnd}-\eqref{eq:CE}
is obtained from the Euler-Lagrange equations \eqref{eq:elnd}
applied to problem \eqref{eq:pcond}:
\begin{equation*}
\begin{cases}
\frac{\square_{\epsilon}}{\square t}\frac{\partial }{\partial
\square_{\epsilon}q} \left({\cal
H}-p\cdot\square_{\epsilon}q\right)
=\frac{\partial}{\partial q}\left({\cal H}-p\cdot\square_{\epsilon}q\right)\\
\frac{\square_{\epsilon}}{\square t}\frac{\partial }{\partial
\square_{\epsilon}u} \left({\cal
H}-p\cdot\square_{\epsilon}q\right) =\frac{\partial}{\partial
u}\left({\cal H}-p\cdot\square_{\epsilon}q\right)\;
\Leftrightarrow\;
\begin{cases}
-\square_{\epsilon}p=\partial_{2} {\cal H}\\
0=\partial_{3} {\cal H}\\
0=\partial_{4} {\cal H}-\square_{\epsilon}q\\
\end{cases}\\
 \frac{\square_{\epsilon}}{\square
t}\frac{\partial }{\partial \square_{\epsilon}p} \left({\cal
H}-p\cdot\square_{\epsilon}q\right)
=\frac{\partial}{\partial p}\left({\cal H}-p\cdot\square_{\epsilon}q\right)\\
\end{cases}
\end{equation*}
\end{proof}

The notion of invariance for problem \eqref{Pond} is defined using
the equivalent augmented problem \eqref{eq:pcond}.

\begin{definition}[\textrm{cf.} Definition~\ref{def:invnd}]
\label{def:invnd-co} The functional \eqref{eq:pcond} is said to be
invariant under the $s$-parameter group of infinitesimal
transformations
\begin{equation}
\label{eq:tinfnd}
\begin{cases}
\bar{t} = t + s\tau(t,q,u,p) + o(s) \, ,\\
\bar{q}(t) = q(t) + s\xi(t,q,u,p) + o(s) \, ,\\
\bar{u}(t) = u(t) + s\varrho(t,q,u,p) + o(s) \, ,\\
\bar{p}(t) = p(t) +s\varsigma(t,q,u,p) + o(s) \, ,\\
\end{cases}
\end{equation}
$\tau$, $\xi$, $\varrho$, $\varsigma \in H^{\beta}$, $\beta\geq
\alpha 1_{[1/2,1]}+(1-\alpha)1_{]0,1/2[}$, if
\begin{multline}
\label{eq:invnd-co} \int_{t_{a}}^{t_{b}}\left[{\cal
H}\left(t,q(t),u(t),p(t)\right)-p(t)\cdot\square_{\epsilon}q(t)\right]dt\\
=\int_{\bar{t}(t_a)}^{\bar{t}(t_b)} \left[{\cal
H}\left(\bar{t},\bar{q}(\bar{t}),\bar{u}(\bar{t}),\bar{p}(\bar{t})\right)
-\bar{p}(\bar{t})\cdot
\square_{\epsilon}\bar{q}(\bar{t})\right]dt
\end{multline}
for any subinterval  $[{t_{a}},{t_{b}}] \subseteq
[a-\epsilon,b+\epsilon]$.
\end{definition}

\begin{definition}[\textrm{cf.} Definition~\ref{def:leicond}]
A function $C(t,q(t),u(t),p(t))$ preserved along any scale
Pontryagin extremal $(q(\cdot),u(\cdot),p(\cdot))$ of problem
\eqref{Pond} is said to be a \emph{scale constant of motion} for
\eqref{Pond}.
\end{definition}

Theorem~\ref{thm:NT:OC} gives a Noether-type theorem for scale
optimal control problems \eqref{Pond}.

\begin{theorem}[Scale Noether's theorem in Hamiltonian form]
\label{thm:NT:OC} If we have invariance in the sense of
Definition~\ref{def:invnd-co}, then
\begin{equation}
\label{eq:tnnd-co} C(t,q(t),u(t),p(t)) ={\cal
H}(t,q(t),u(t),p(t))\tau-p(t)\cdot\xi
\end{equation}
is a scale constant of motion for \eqref{Pond}.
\end{theorem}

\begin{proof}
The scale constant of motion \eqref{eq:tnnd-co} is obtained by
applying Theorem~\ref{theo:tnnd} to problem \eqref{eq:pcond}.
\end{proof}

\begin{remark}
For the scale problem of the calculus of variations \eqref{Pe} the
Hamiltonian \eqref{eq:Hnd} takes the form ${\cal H} = L + p \cdot
u$, with $u = \square_{\epsilon}{q}(t)$ and $p = - \partial_3 L$
(\textrm{cf.} Remark~\ref{rem:cp:CV:EP:EEL}). In this case the
scale constant of motion \eqref{eq:tnnd-co} reduces to
\eqref{eq:tnnd}.
\end{remark}


\section{Application: Scale Constants of Motion for Schr\"{o}dinger Equations}

In \cite[\S 5]{CD:Cresson:2005} some fractional variants of the
Schr\"{o}dinger equation, with particular interest in quantum
mechanics, are studied. It is proved that under certain
conditions, solutions of both linear and nonlinear
Schr\"{o}dinger's equations coincide with the extremals of certain
functionals \eqref{Pe} of the scale calculus of variations. In
this section we use our Noether theorem to find scale constants of
motion for the problems studied in \cite[\S 5]{CD:Cresson:2005}.
In all the examples we use the computer program
\cite{GouveiaTorresCMAM,GouveiaTorresRocha06}\footnote{The
software is available from the Maple Application Centre at
[\url{http://www.maplesoft.com/applications/app_center_view.aspx?AID=1983}].}
to compute the symmetries (\textrm{i.e.} the invariance
transformations \eqref{eq:tinf:sd}).

\begin{example}
Let us consider the following nonlinear Schr\"{o}dinger equation:
 \begin{equation}
 \label{eq:sch}
 2i\gamma m\left[-\frac{1}{\Psi}\left(\frac{\partial \Psi}{\partial q}\right)^2
 \left(i\gamma+\frac{a_\epsilon(t)}{2}\right)+\frac{\partial \Psi}{\partial t}
 +\frac{a_\epsilon(t)}{2}\frac{\partial^2 \Psi}{\partial
 q^2}\right]
 =\left(U(q)+\alpha(q)\right)\Psi
\end{equation}
where $m>0$, $\gamma\in\mathbb{R}$,
$U:\mathbb{R}\longmapsto\mathbb{R}$, $\alpha(q)$ is an arbitrary
continuous function, $q(t)\in H^{1/2}$, $\Psi(t,q)$ satisfy the
condition
\begin{equation*}
\frac{\square_\epsilon q(t)}{\square t}=-i2\gamma\frac{\partial
\ln(\Psi(t,q))}{\partial q}
\end{equation*}
and $a_\epsilon:\mathbb{R}\longmapsto\mathbb{C}$ is given by
\begin{equation*}
a_\epsilon(t)=\frac{1}{2}\left[\left(\left(\Delta_\epsilon^+q(t)\right)^2-
\left(\Delta_\epsilon^-q(t)\right)^2\right)-i\left(\left(\Delta_\epsilon^+q(t)\right)^2+
\left(\Delta_\epsilon^-q(t)\right)^2\right)\right]\, .
\end{equation*}
It is shown in \cite[Theorem~5.1]{CD:Cresson:2005} that the
solutions $q(t)$ of \eqref{eq:sch} coincide with the
Euler-Lagrange extremals of functional \eqref{Pe} with the
Lagrangian
 \begin{equation*}
  L(t,q(t),\square_\epsilon q(t))=\frac{1}{2}m\left(\square_\epsilon
 q(t)\right)^2+U(q)\, .
 \end{equation*}
The functional
\begin{equation*}
I[q(\cdot)] =\frac{1}{2}\int_a^b
\left[m\left(-i2\gamma\frac{\partial \ln(\Psi(t,q))}{\partial
q}\right)^2 +2U(q)\right]dt
\end{equation*}
is invariant in the sense of Definition~\ref{def:invnd} under the
symmetries $(\tau,\xi)=(c,0)$, where $c$ is an arbitrary constant.
It follows from our Theorem~\ref{theo:tnnd} that
\begin{equation}
\label{eq:scm:Ex1} -2m\left(\gamma\frac{\partial
\ln(\Psi(t,q))}{\partial q}\right)^2+U(q)
\end{equation}
is a scale constant of motion: \eqref{eq:scm:Ex1} is preserved
along all the solutions $q(t)$ of the nonlinear Schr\"{o}dinger
equation \eqref{eq:sch}.
\end{example}

\begin{example}
We now consider the following linear Schr\"{o}dinger's equation:
\begin{equation}
\label{eq:sch1} i\bar{h}\frac{\partial \Psi}{\partial
t}+\frac{\bar{h}^2}{2m}\frac{\partial^2 \Psi}{\partial
q^2}=U(q)\Psi
\end{equation}
where $\bar{h}=\frac{h}{2\pi}$, $m>0$,
$U:\mathbb{R}\longmapsto\mathbb{R}$, $\Psi(t,q)$ satisfy
\begin{equation*}
\frac{\square_\epsilon q(t)}{\square
t}=-i\frac{\bar{h}}{m}\frac{\partial \ln(\Psi(t,q))}{\partial q}
\end{equation*}
and $q(t)\in H^{1/2}$ are such that
\begin{equation*}
\frac{1}{2}\left[\left(\left(\Delta_\epsilon^+q(t)\right)^2-
\left(\Delta_\epsilon^-q(t)\right)^2\right)-i\left(\left(\Delta_\epsilon^+q(t)\right)^2+
\left(\Delta_\epsilon^-q(t)\right)^2\right)\right]=-i\frac{\bar{h}}{m}
\, .
\end{equation*}
In \cite[Theorem 5.2]{CD:Cresson:2005} it is proved that the
solutions of \eqref{eq:sch1} coincide with Euler-Lagrange
extremals of functional \eqref{Pe} with Lagrangian
 \begin{equation*}
  L(t,q(t),\square_\epsilon q(t))=\frac{1}{2}m\left(\square_\epsilon
 q(t)\right)^2+U(q)\, .
 \end{equation*}
It happens that functional
\begin{gather*}
I[q(\cdot)] =\frac{1}{2}\int_a^b
\left[m\left(-i\frac{\bar{h}}{m}\frac{\partial
\ln(\Psi(t,q))}{\partial q}\right)^2 +2U(q)\right]dt
\end{gather*}
is invariant in the sense of Definition~\ref{def:invnd} under the
symmetries $(\tau,\xi)=(c,0)$, where $c$ is an arbitrary constant.
It follows from our Theorem~\ref{theo:tnnd} that
\begin{equation}
\label{eq:scm:Ex2} -\frac{1}{2m}\left(\bar{h}\frac{\partial
\ln(\Psi(t,q))}{\partial
q}\right)^2+U(q)=-\frac{1}{8m}\left(\frac{h}{\pi}\frac{\partial
\ln(\Psi(t,q))}{\partial q}\right)^2+U(q)
\end{equation}
is a scale constant of motion: expression \eqref{eq:scm:Ex2} is
constant along all the solutions $q(t)$ of the linear
Schr\"{o}dinger's equation \eqref{eq:sch1}.
\end{example}


\section*{Acknowledgements}

Delfim F. M. Torres is grateful to the University of Pau, Academy
of Bordeaux, for a one-month position of Invited Professor at the
Laboratoire de Math\'{e}matiques Appliqu\'{e}es, Universit\'{e} de
Pau et des Pays de l'Adour, France, where the present work was
finished. The hospitality and the good working conditions at Pau
are very acknowledged.



\end{document}